\documentclass{article}
\usepackage{graphicx,float}
\usepackage{t1enc}
\usepackage[latin1]{inputenc}
\usepackage{mathrsfs,xspace,theorem}
\usepackage{amsmath,amssymb}
\usepackage{textcomp}
\usepackage{authblk} 
\theorembodyfont{\rmfamily}
\floatstyle{ruled}
\newfloat{Algorithm}{htpb}{alg}

\newcounter{prgline}

\newcounter{noqed}
\newcommand{\qed}{ \ifmmode\mbox{
}\fi\rule[-.05em]{.3em}{.7em}\setcounter{noqed}{0}}
\newenvironment{proof}[1][{}]{\noindent{\bf Proof#1.
}\setcounter{noqed}{1}}{\ifnum\value{noqed}=1\qed\fi\par\medskip}

\def\..{\,\mathpunct{\ldotp\ldotp}} % Middle stuff for intervals. Usage: \..

\newcommand{\A}{\mathscr A}

\newcommand{\url}{\cite{myurl}}

\newtheorem{lemma}{Lemma} 

\newtheorem{theorem}{Theorem}
\newtheorem{corollary}{Corollary}

\renewcommand{\emptyset}{\varnothing}
\renewcommand{\epsilon}{\varepsilon}
\renewcommand{\phi}{\varphi}

\begin{document} 
%\linenumbers
%\sloppy
\title{Worst case efficient single and multiple string matching in the Word-RAM model
\thanks{This work is supported by the french ANR project MAPPI.}
\thanks{A preliminary version of this paper was presented at the 21st International Workshop on Combinatorial Algorithms(IWOCA)  London 2010}
}   

\author{Djamal Belazzougui}
\affil{LIAFA, Univ. Paris Diderot - Paris 7, 75205 Paris Cedex 13, France dbelaz@liafa.jussieu.fr}
%\thanks{Research partially supported by ANR project MAPPI}
\bibliographystyle{abbrv}
\maketitle
\begin{abstract}

In this paper, we explore worst-case solutions for the problems of single and multiple matching on strings in the word RAM model with word length $w$. 
In the first problem, we have to build a data structure based on a pattern $p$ of length $m$ over an alphabet of size $\sigma$ such that we can answer to the following query: given a text $T$ of length $n$, where each character is encoded using $\log\sigma$ bits return the positions of all the occurrences of $p$ in $T$ (in the following we refer by $occ$ to the number of reported occurrences). For the multi-pattern matching problem we have a set $S$ of $d$ patterns of total length $m$ and a query on a text $T$ consists in finding all positions of all occurrences in $T$ of the patterns in $S$. As each character of the text is encoded using $\log\sigma$ bits and we can read $w$ bits in constant time in the RAM model, we assume that we can read up to $\Theta(w/\log\sigma)$ consecutive characters of the text in one time step. This implies that the fastest possible query time for both problems is $O(n\frac{\log\sigma}{w}+occ)$. In this paper we present several different results for both problems which come close to that best possible query time. 
We first present two different linear space data structures for the first and second problem: the first one answers to single pattern matching queries in time $O(n(\frac{1}{m}+\frac{\log\sigma}{w})+occ)$ while the second one answers to multiple pattern matching queries to $O(n(\frac{\log d+\log y+\log\log m}{y}+\frac{\log \sigma}{w})+occ)$ where $y$ is the length of the shortest pattern. We then show how a simple application of the four russian technique permits to get data structures with query times independent of the length of the shortest pattern (the length of the only pattern in case of single string matching) at the expense of using more space.

% at the expense of using additional space $t$, resulting in query times $O(\frac{n}{\log_\sigma t}+occ)$ and  $O(n\frac{\log d+\log\log_\sigma t+\log\log m}{\log_\sigma t}+occ)$ respectively for the first and second data structure. 

%In this paper, we present three results. %Our first result improves the worst-case performance in the RAM model using a condition on the length of the patterns which was previously used to get improved average-case performance. 
%The first result is a data structure which uses linear space and answers to  
%For sufficiently long $m$, the query time tends to the optimal $O(n\frac{\log \sigma}{w}+occ)$. 
%multiple pattern matching queries in time  (single pattern matching queries are answered in time $O(n(\frac{\log m}{m}+\frac{\log \sigma}{w})+occ))$). %The query time tends to the optimal $O(n\frac{\log \sigma}{w}+occ)$ as $y$ tends to infinity. 
%Our second result is a variant of the first result which uses the four Russian technique to remove the dependence on the shortest pattern length at the expense of using an additional space $t$. It answers to multiple pattern matching queries in time  using $O(m+t)$ words of space. The third result presented in this paper is a linear space data structure able to answer to single-matching queries in time . 
\end{abstract}
\section{Introduction}
The problems of string pattern matching and multiple string pattern matching are classical algorithmic problems in the area of pattern matching. In the multiple string matching problem, we have to preprocess a dictionary of $d$ strings of total length $m$ characters over an alphabet of size $\sigma$ so that we can answer to the following query: given any text of length $n$, find all occurrences in the text of any of the $d$ strings. In the case of single string matching, we simply have $d=1$.
\\The textbook solutions for the two problems are the Knuth-Morris-Pratt~\cite{KMP77} (KMP for short) automaton for the single string matching problem and the Aho-Corasick~\cite{AC75} automaton (AC for short) for the multiple string matching problem. The AC automaton is actually a generalization of the KMP automaton. 
Both algorithms achieve $O(n+occ)$ query time (where $occ$ denotes the number of reported occurrences) using $O(m\log m)$ bits of space\footnote{In this paper we quantify the space usage in bits rather than in words as is usual in other papers} (both automatons are encoded using $O(m)$ pointers occupying $\log m$ bits each). The query time of both algorithms is in fact optimal if the matching is restricted to read all the characters of the text one by one. However as it was noticed in may previous works, in many cases it is actually possible to avoid reading all the characters of the text and hence achieve a better performance. This stems from the fact that by reading some characters at certain positions in the text, one could conclude whether a match is possible or not without the need to read all the characters. This has led to various algorithms with so-called sublinear query time assuming that the characters of the patterns and/or the text are drawn from some random distribution. 
The first algorithm which exploited that fact was the Boyer-Moore algorithm~\cite{BM77}. Subsequently other algorithms with provably average-optimal performance were devised. Most  notably the BDM and BNDM for single string matching and the multi-BDM~\cite{CR94,CCGJLPR94} and multi-BNDM ~\cite{NR98} for multiple string matching. Those algorithms achieve $O(n\frac{\log m}{m\log\sigma}+occ)$ time for single string matching (which is optimal according to the lower bound in~\cite{YAO79}) and $O(n\frac{\log d+\log y}{y\log\sigma}+occ)$ time for multiple string matching, where $y$ is the length of the shortest string in the set. 
%Actually, it was proved that the best average time one could hope for (assuming only a single character can be read at each time step) is $O(n\frac{\log m}{m\log\sigma}+occ)$ for single string matching and $O(n\frac{\log y+\log d}{y\log\sigma}+occ)$ for multiple string matching, where $y$ is the length of the shortest string in the set. 
Still in the worst case those algorithms may have to read all the text characters and thus have $\Omega(n+occ)$ query time (actually many of those algorithms have an even worse query time in the worst-case, namely $\Omega(nm+occ)$).  
\\A general trend has appeared in the last two decades when many papers have appeared trying to exploit the power of the word RAM model to speed-up and/or reduce the space requirement of classical algorithms and data structures. In this model, the computer operates on words of length $w$ and usual arithmetic and logic operations on the words all take one unit of time. 
\\
In this paper we focus on the worst-case bounds in the RAM model with word length $w$. That is we try to improve on the KMP and AC in the RAM model assuming that we have to read all the characters of the text which are assumed to be stored in a contiguous area in memory using $\log\sigma$ bits per characters. %If we consider string algorithms in the RAM model then, we can assume that the strings are encoded using $\log\sigma$ bits per character and that the processor operates on words of length $w$ bits 
That means that it is possible to read $\Theta(w/\log\sigma)$ consecutive characters of the text in $O(1)$ time. Thus given a text of length $n$ characters, an optimal algorithm should spend $O(n\frac{\log\sigma}{w}+occ)$ time to report all the occurrence of matching patterns in the text. The main result of this paper is a worst case efficient algorithm whose performance is essentially the addition of a term similar to the average optimal time presented above
%(apart from a $\log\sigma$ factor) 
plus the time necessary to read all the characters of the text in the RAM model. Unlike many other papers, we only assume that $w=\Omega(\log(n+m))$, and not necessarily that $w=\Theta(\log(n+m))$. That is we only assume that a pointer to the manipulated data (the text and the patterns), fit in a memory word but the word length $w$ can be arbitrarily larger than $\log m$ or $\log n$. This assumption makes it possible to state time bounds which are independent of $m$ and $n$, implying larger speedups for small values of $m$ and $n$. 
\\
In his paper Fredriksson presents a general approach~\cite{F02} which can be applied to speed-up many pattern matching algorithms. This approach which is based on the notion of super-alphabet relies on the use of tabulation (four russian technique). If this approach is applied to our problems of single and multiple string matching queries, given an available precomputed space $t$, we can get a $\log_\sigma (t/m)$ factor speedup. 
In his paper~\cite{B09}, Bille presented a more space efficient method for single string matching queries which accelerates the KMP algorithm to answer to queries in time $O(\frac{n}{\log_\sigma n}+occ)$ using $O(n^{\epsilon}+m\log m)$ bits of space for any constant $\epsilon$ such that $0<\epsilon<1$. 
%This index is based on KMP~\cite{KMP77} algorithm. 
More generally, the algorithm can be tuned to use an additional amount $t$ of tabulation space in order to provide a $\log_\sigma t$ factor speedup. 
\\At the end of his paper, Bille asked two questions: the first one was whether it is possible to get an acceleration proportional to the machine word length $w$ (instead of $\log n$ or $\log t$) using linear space only. The second one was whether it is possible to obtain similar results for the multiple string matching problem. We give partial answers to both questions. Namely, we prove the following two results:
\begin{enumerate}
\item Our first result states that for $d$ strings of minimal length $y$, we can construct an index which occupies linear space and answers to queries in time $O(n(\frac{\log d+\log y+\log\log m}{y}+\frac{\log\sigma}{w})+occ)$. This result implies that we can get a speedup factor $\frac{w}{(\log d+\log w)\log\sigma}$ if $y\geq \frac{w}{\log\sigma}$ and get the optimal speedup factor $\frac{w}{\log\sigma}$ if $y\geq (\log d+\log w)\frac{w}{\log\sigma}$. 
\item Our second result implies that for $d$ patterns of arbitrary lengths and an additional $t$ bits of memory, we can obtain a factor $\frac{\log_\sigma t}{\log d+\log\log_\sigma t+\log\log m}$ speedup using $O(m\log m+t)$ bits of memory.
\end{enumerate}
Our first result compares favorably to Bille's and Fredriksson approaches as it does not use any additional tabulation space. In order to obtain any significant speedup, the algorithms of Bille and Fredriksson require a substantial amount of space $t$ which is not guaranteed to be available. Even if such an amount of space was available, the algorithm could run much slower in case $m\ll t$ as modern hardware is made of memory hierarchies, where random access to large tables which do not fit in the fast levels of the hierarchy might be much slower than access to small data which fit in faster levels of the hierarchy. 
\\ 
Our second result is useful in case the shortest string is very short and thus, the first result do not provide any speedup. The result is slightly less efficient than that of Bille for single string matching, being a factor $\log\log_\sigma t+\log\log m$ slower (compared to the $\log_\sigma t$ speedup of Bille's algorithm). However, our second result efficiently extends to multiple string matching queries, while Bille's algorithms seems not to be easily extensible to multiple string matching queries.% unless building $d$ independent data structure for a set of $d$ strings.
% and consequently paying an additional factor $\Theta(d)$ in query time giving a $\frac{\log_\sigma t}{d}$ speedup factor compared to the $\frac{\log_\sigma t}{\log d+\log\log_\sigma t+\log\log m}$ speedup of our algorithm. Thus for $d=\Theta(\log_\sigma t)$, the straightforward extension of Bille's algorithm gives no speedup while our algorithm gives a $\frac{\log_\sigma t}{\log\log_\sigma t+\log\log m}$ speedup. 
\\
The third and fourth results in this paper are concerned with single string matching, where we can have solutions with a better query time than what can be obtained by using the first and second result for matching a single pattern. In particular our results imply the following:
\begin{enumerate}
\item Given a single pattern $p$ of length $m$, we can construct a data structure which occupies a linear space and which can find all $occ$ occurrences of $p$ in any text of length $n$ in time $O(n(\frac{1}{m}+\frac{\log\sigma}{w})+occ)$. This implies that we can get optimal query time  $O(n\frac{\log\sigma}{w}+occ)$ as long as $m\geq \frac{w}{\log\sigma}$.
\item For a single string of length $m$ and having some additional $t$ bits of space, we can build a data structure which occupies $O(m\log m+t)$ bits of space such that all the $occ$ occurrences of $p$ in any text of length $n$ are reported in time $O(n/\log_\sigma t+occ)$
\end{enumerate}

In a recent work~\cite{B10a}, we have tried to use the power of the RAM model to improve the space used by the AC representation to the optimal (up to a constant factor) $O(m\log\sigma)$ bits instead of $O(m\log m)$ bits of the original representation, while maintaining the same query time. In this paper, we attempt to do the converse. That is, we try to use the power of the RAM model to improve the query time of the AC automaton while using the same space as the original representation. 
%Thus our result is complementary to that of~\cite{B10a}. An interesting open problem is then to try to improve both the query time and the space usage of the AC automaton. 
\\
We emphasize that our results are mostly theoretical in nature. The constants in space usage and query time of our data structures seem rather large. Moreover, in practice average efficient algorithms which have been tuned for years are likely to behave much better than any worst-case efficient algorithm. For example, for DNA matching, it was noted that DNA sequences encountered in practice are rather random and hence average-efficient algorithms tend to perform extremely well for matching in DNA sequences (see ~\cite{RSKKT09} for example). 
\section{Outline of the results}
\subsection{Problem definition, notation and preliminaries}
In this paper, we aim at addressing two problems: the single string pattern matching and the multiple string pattern matching problems. In the single string pattern matching problem we have to build a data structure on a single pattern (string) of length $m$ over an alphabet of size $\sigma\leq m$~\footnote{Our results also apply to the $\sigma\geq m$. The only change is in space bounds in which the term $m\log m$ should be replaced by $m\log\sigma$}. In the multiple string pattern matching problem, we have a set $S$ of $d$ patterns of total length $m$ characters where each character is drawn from an alphabet of size $\sigma\leq m$. In the first problem, we have to identify all occurrences of the pattern in a text $T$ of length $n$. In the second problem, we have to identify all occurrences of any of the $d$ patterns. 
\\
In this paper, we assume a unit-cost RAM model with word length $w$, and assume that $w=\Omega(\log m+\log n)$. However $w$ could be arbitrarily larger than $\log m$ or $\log n$. We assume that the patterns and the text are drawn from the same alphabet $\Sigma$ of size $\sigma\leq m$. We assume that all usual RAM operations (multiplications, additions, divisions, shifts, etc...) take one unit of time. 
\\For any string $x$ we denote by $x[i,j]$ (or $x[i..j]$) the substring of $x$ which begins at position $i$ and ends at position $j$ in the string $x$. 
For any integer $m$ we note by $\log m$ the integer number $\lceil \log_2 m\rceil$.
\\
In the paper we make use of two kinds of ordering on the strings: the prefix lexicographic order which is the standard lexicographic ordering (strings are compared right-to-left) and the suffix-lexicographic order which is defined in the same way as prefix lexicographic, but in which string are compared left-to-right instead of right-to-left. The second ordering can be thought as if we write the strings in reverse before comparing them. Unless otherwise stated, string lengths are expressed in terms of number of characters. We make use of the fixed integer bit concatenation operator $(\cdot)$ which operates on fixed length integers, where $z=x\cdot y$ means that $z$ is the integer whose bit representation consists in the concatenation of the bits of the integers $x$ as most significant bits followed by the bits of the integer $y$ as least significant bits. 
We define the function $sucount_X(s)$, which returns the number of elements of a set $X$ which have a string $s$ as a suffix. Likewise we define the function $prcount_X(s)$, which returns the number of elements of a set $X$ which have a string $s$ as a prefix. 
We also define two other functions $surank_X(s)$ and $prrank_X(s)$ as the functions which return the number of elements of a set $X$ which precede the string $s$ in suffix and prefix lexicographic orders respectively.

\subsection{Results} 
The results of this paper are summarized by the following two theorems: 
\begin{theorem}
\label{theorem1}
Given a set $S$ of $d$ strings of total length $m$, where the shortest string is of length $y$, we can build a data structure of size $O(m\log m)$ bits such that given any text $T$ of length $n$, we can find all occurrences of strings of $S$ in $T$ in time $O(n(\frac{\log d+\log y+\log\log m}{y}+\frac{\log\sigma}{w})+occ)$.
\end{theorem}
The theorem give us the following interesting corollaries:
%\begin{corollary}
%\label{corollary1}
%Given a string $p$ of length $m$, we can build a data structure of size $O(m\log m)$ bits of space such that given any text $T$ of length $n$, we can find all $occ$ occurrences of $p$ in $T$ in time $O(n(\frac{\log m}{m}+\frac{\log\sigma}{w})+occ)$.
%\end{corollary}
%By setting $d=1$ in the theorem, we get the following corollary for single string matching: 
For multiple string matching, we have the following two corollaries:
\begin{corollary}
\label{corollary2}
Given a set $S$ of $d$ strings of total length $m$ where each string is of length at least $\frac{w}{\log\sigma}$ characters, we can build a data structure of size $O(m\log m)$ bits of space such that given any text $T$ of length $n$, we can find all occurrences of strings of $S$ in $T$ in time $O(n\frac{(\log d+\log w)\log\sigma}{w}+occ)$.
\end{corollary}
%Note that in the case of $\sigma=O(1)$ for DNA alphabet for example, corollary~\ref{corollary1} implies that we can look for $d$ fragments in a sequence of length $n$ in time $O(n\frac{(\log d+\log w)\log\sigma}{w})$ if each fragment is of length at least $w$ characters. This is a factor $w/(\log d+\log w)$ speedup compared to the standard AC automaton. 
For the case of even larger minimal length, we can get optimal query time :
%the following corollary %which improves on corollary~\ref{corollary2}
%in case of sufficiently long strings:
%, if we set $y=(\log d+\log w)\frac{w}{\log\sigma}$ in theorem~\ref{theorem1}:
\begin{corollary}
\label{corollary3}
Given a set $S$ of $d$ strings of total length $m$ where each string is of length at least $(\log d+\log w)\frac{w}{\log\sigma}$ characters, we can build a data structure occupying $O(m\log m)$ bits of space such that given any text $T$ of length $n$, we can find all occurrences of strings of $S$ in $T$ in the optimal $O(n\frac{\log\sigma}{w}+occ)$ time.
\end{corollary}  
The dependence of the bounds in theorem~\ref{theorem1} and its corollaries on minimal patterns lengths is not unusual. 
%Similar restriction is made in 
This dependence exists also in average-optimal algorithms like BDM, BNDM and their multiple patterns variants~\cite{CR94,CCGJLPR94,NR98}. Those  algorithms achieve a $\frac{y\log\sigma}{\log d+\log y}$ speedup factor on average requiring that the strings are of minimal length $y$. 
%It seems very difficult to obtain a speedup factor close to the word length $w$ without making that restriction and without using any tabulating space. Fredriksson has shown how to achieve $y$ but using space precomputed table with $\Theta(m(y\log\sigma)^\epsilon)$. 
Our query time is the addition of a term which represents the time necessary to read all the characters of text in the RAM model and a term which is similar to the query time of the  average optimal algorithms. 
\\
We also show a variation of the first theorem which uses the four russian technique and which will mostly be useful in case the minimal length is too short: 
\begin{theorem}
\label{theorem2}
Given a set $S$ of $d$ strings of total length $m$ and an integer parameter $\alpha$, we can build a data structure occupying $O(m\log m+\sigma^\alpha\log^2\alpha\log m)$ bits of space such that given any text $T$ of length $n$, we can find all $occ$ occurrences of strings of $S$ in $T$ in time $O(n\frac{\log d+\log s+\log\log m}{\alpha}+occ)$.
\end{theorem}
The theorem could be interpreted in the following way: having some additional amount $t$ of available memory space, we can achieve a speedup factor $\frac{\alpha}{\log d+\log \alpha}$ for $\alpha=\log_\sigma t$ using a data structure which occupies $O(m\log m+t)$ bits of space.  
\\ The theorem gives us two interesting corollaries which depend on the relation between $m$ and $n$. In the case where $n\geq m$, by setting $t=n^{\epsilon}$ for any $0<\epsilon<1$, we get the following corollary: 

\begin{corollary}
\label{corollary4}
Given a set $S$ of $d$ strings of total length $m$, we can build a data structure occupying $O(m\log m+n^{\epsilon})$ bits of space such that given any text $T$ of length $n$, we can find all occurrences of strings of $S$ in $T$ in time $O(n\frac{\log d+\log\log_\sigma n+\log\log m}{\log_\sigma n}+occ)$, where $\epsilon$ is any constant such that $0<\epsilon<1$.
\end{corollary}
%Note that in this corollary we can assume than $w=\Theta(\log n)$ and thus  the space 

In the case $m\geq n$ %(which is only possible in the case of multiple string matching
we can get a better speedup by setting $t=m$:
\begin{corollary}
\label{corollary5}
Given a set $S$ of $d$ strings of total length $m$, we can build a data structure occupying $O(m\log m)$ bits of space such that given any text $T$ of length $n$, we can find all occurrences of strings of $S$ in time $O(n\frac{\log d+\log\log m}{\log_\sigma m}+occ)$.
\end{corollary}
We note that in the case $d=1$, the result of corollary~\ref{corollary2} is worse by a factor $\log\log_\sigma n+\log\log m$ than that of Bille which achieves a query time of $O(\frac{n}{\log_\sigma t}+occ)$. However the result of Bille does not extend naturally to $d\geq 1$. The straightforward way of extending Bille's algorithm is to build $d$ data structures and to match the text against all the data structures in parallel. This however would give a running time of $O(n\frac{d}{\log_\sigma n}+occ)$ which is worse than our running time $O(n\frac{\log d+\log\log_\sigma n+\log\log m}{\log_\sigma n}+occ)$ which is linear in $\log d$ rather than $d$.

As of the technique of Fredriksson, in order to obtain query time $O(\frac{n}{\alpha}+occ)$, it needs to use at least space $\Omega(m\sigma^\alpha)$ which can be too much in case $\alpha$ is too large. 
%Note that in this case the space usage is prohibitive when compared with ours in case where $m$ is rather large.
%The two colrollaries of theorem~\ref{theorem2} can be used in the following way : we can build the data structure of theorem~\ref{corollary4} and then when receiving a text of lentgh $n$, we check whether $n\leq m^2$ in which case, we use the data structure of corollary~\ref{corollary5} for which we allocate an $n^\epsilon$ additional space (this space is sublinear compared to $n$). 
In the case of single pattern matching, we can even get a stronger results as we prove the following theorem: 
\begin{theorem}
\label{theorem3}
Given a string $p$ of length $m$, we can build a data structure occupying $O(m\log m)$ bits of space such that given any text $T$ of length $n$, we can find all $occ$ occurrences of the string $s$ in $T$ in time $O(n(\frac{1}{m}+\frac{\log\sigma}{w})+occ)$.
\end{theorem}
An important implication of this theorem is that single pattern matching in optimal time $O(n\frac{\log\sigma}{w}+occ)$ is possible for strings of length $m\geq \frac{w}{\log\sigma}$. 
\\ Similarly to the case of~\ref{theorem1}, we can use the four russian technique to improve the result of~\ref{theorem2} in case $m$ is too short: 
\begin{theorem}
\label{theorem4}
Given a string $p$ of length $m$ we can build a data structure occupying $O(m\log m+\sigma^\alpha\log \alpha)$ bits of space such that given any text $T$ of length $n$, we can find all $occ$ occurrences of strings of $S$ in $T$ in time $O(n\frac{1}{\alpha}+occ)$. 
\end{theorem}
This last theorem matches the result achieved in Bille's  algorithm. 
\section{Components}

Before we present the details of our main results we first present the main tools and components which are to be used in our solutions. In particular we will make use of several data structures and operations which exploit the power of the word-RAM model. We first describe some basic operations which will be explicitly used for implementing our algorithms. Then we describe some classical geometric and string processing oriented data structures which will be used as black-box components in our data structures.  
\subsection{Bit parallel string processing}
Before we describe the basic bit-parallel operations, we first define how the characters are packed in words. We assume that the pattern and the text are packed in a similar way. Each character is encoded using $\log\sigma$ bits. The text $T$ is thus encoded using a bit array $B_T$ which occupies $m\log\sigma$ bits which is $\lceil m\log\sigma/w\rceil$ words. We thus assume that have a representation of the text $T$ which fits in a word array $W_T$ \footnote{Notice that when $w$ is not multiple of $\log\sigma$ , a character could span a boundary between two consecutive words}. An important technical point is about the endianness, that is the way the bits are ordered in a word which influences the way the characters are packed in memory. We basically have two possibilities: either the bits in a word are ordered from the least to the most significant (little endian) or the converse (big endian). Here we illustrate how a particular character $T[i]$ of the text is extracted. We only present the first case as (little endian) as the latter can easily be deduced from the former: 
\begin{enumerate}
\item First compute $i_0=(i\log\sigma)\bmod w$.
\item Then read the two words $W_0=W_T[\lfloor i\log\sigma/w\rfloor]$ and $W_1=W_T[\lfloor (i+1)\log\sigma/w\rfloor]$.
\item At last we distinguish two cases: 
\begin{itemize}
\item If $\lfloor i\log\sigma/w\rfloor=\lfloor (i+1)\log\sigma/w\rfloor$ (the character $i$ does not span two consecutive words), then return $(W_0\gg i_0)\bmod \sigma$ 
\item Otherwise (the character spans the two consecutive words $W_0$ and $W_1$) we return $(W_0\gg i_0)+(W_1 \bmod 2^{\log\sigma-(w-i_0)})$.
\end{itemize}
\end{enumerate}
It can easily be seen that the extraction of a character can be done in constant tome. However, in general we will want to make operations on groups of characters instead of manipulating characters one bye one. This permits to get much faster operations on strings. In particular we will makes use of the following lemma whose proof is omitted and which can easily be implemented using standard bit-parallel instructions. 
\begin{lemma}
Given two strings of lengths $m<\frac{w}{\log\sigma}$ bits, one can compare them (for equality) in $O(1)$ time using bit-parallelism. Moreover, given two strings of length $m$, one can compare them in time $O(m\frac{\log\sigma}{w})$. 
\end{lemma}
\paragraph{MSB and LSB operations}
Our solutions for single string matching uses the special instruction $MSB(x)$ which returns the most significant bit set in a word and similarly $LSB(x)$ which returns the least significant set bit in a word. 
%We have to try to find how to implement the operation using multiplication only. Arguably, according to the paper by Raman and Thorup this operation is non standard but still is in $AC^0$, it can be done in C by first loading the number as a floating point number and then computing the exponent. 
Those two operations can be simulated in constant time using classical RAM operations (see~\cite{AHNR98,FW93,Brodnik93}).
\begin{lemma}
The two functions $MSB(x)$ and $LSB(x)$ 
%which return the most and least significant bit of a bit-string $x$ respectively 
can be implemented in $O(1)$ time provided that the bit-string $x$ is of length $O(w)$ bits.
\end{lemma}

\paragraph{Longest repetition matching}
We will make use of the following tool: given a string $p$ of length $m$ and a string $s$ of length $n>m$ where both strings are over the same alphabet of size $\sigma$, we would wish to have the following two operations: 
\begin{enumerate}
\item Longest prefix repetition matching: find the largest $i$ such that $p^i$ ($p$ repeated $i$ times) is a prefix of $s$.
\item Longest suffix repetition matching: find the largest $i$ such that $p^i$ is a suffix of $s$.
\end{enumerate}
We argue that both operations can be done in $O(n\frac{\log\sigma}{w})$. First consider the computation of Longest prefix repetition of a string $p$ of length $m$ into a string $s$ of length $n$. We have two cases:
\begin{enumerate}
\item Suppose that $m\log\sigma\geq w/2$. In this case, it suffice to compare successively $s[mi,m(i+1)-1]$ with $p$ for increasing values of $i$ until we reach $i=\lfloor\frac{n}{m}\rfloor$ or find a mismatch. Each comparison takes time $O(1)$ and thus the whole operation takes at most $O(n\frac{\log\sigma}{w})$ time. 
\item Suppose that $m\log\sigma<w/2$, in this case we first compute $k=\lfloor\frac{w}{m\log\sigma}\rfloor$ and then compute $p'=p^k$ and note that $w/2<m'\log\sigma\leq w$. Now we first compare $s[m'j,m'(j+1)-1]$ with $p'$ for increasing values of $j$ until we reach $j=\lfloor\frac{n}{m'}\rfloor$ or find a mismatch. Clearly this step takes time $O(n\frac{\log\sigma}{w})$ also. Now, we have determined that $jk\leq i<j(k+1)$. In the final step we compute $q=s[m'j,m'(j+1)-1]$ and finally $r=(q\oplus p')$ (where $\oplus$ denotes the xor operator) and let $t=LSB(r)$ (or $t=MSB$ depending on the \emph{endiannes} or the way the processors orders the bits in its words). Now clearly, $t$ is the position of the first bit in which $p'$ and $q$ differ. It is clear that the first character in which $p'$ and $q$ differ, is precisely character number $\lfloor t/\log\sigma\rfloor$. From there we deduce that $i=jk+\lfloor t/\log\sigma\rfloor$. The computation of the $LSB$ and the xor operator both take constant time. 
\end{enumerate}
The computation of the longest suffix repetition is symmetric to the computation of the longest prefix repetition except that we use $MSB$ operation instead of $LSB$ or vice-versa depending on the endiannes. 
%\\Thus we have proved the following lemma:
\begin{lemma}
Given a string $p$ of length $m$ and a string $s$ of length $n$ where $n>m$ the longest prefix (and suffix) repetition of $s$ in $p$ can be found in time $O(m\frac{\log\sigma}{w})$.
\end{lemma}
\subsection{Data structures components}
For our results we will use several classical data structures which are illustrated with the following lemmata:
\begin{lemma}\cite{W83}
\label{lemma1}
Given a collection of $n$ intervals over universe $U$ where for any two intervals $s_1$ and $s_2$ we have either $s_1\cap s_2=s_1$, $s_1\cap s_2=s_2$ or $s_1\cap s_2=\emptyset$ (for any two intervals either one is included in the other or the two intervals are disjoint). We can build a data structure which uses $O(n\log n)$ bits of space such that for any point $x$, we can determine the interval which most tightly encloses $x$ in $O(n\log\log n)$ time (the smallest interval which encloses $x$). 
\end{lemma}
For implementing the lemma, we store the set of interval endpoints in a predecessor data structure, namely the Willard's y-fast trie ~\cite{W83} which is a linear space version of the Van Emde Boas tree~\cite{BKZ77}. Then those points divide the universe of size $U$ into $2n+1$ segments and each segment will point to the interval which most tightly encloses the segment. Then a predecessor query will point to the segment which in turn points to the relevant interval. This problem can be thought as a restricted $1D$ stabbing problem (in the general problem we do not have the condition that for any two intervals either one is included in the other or the two intervals are disjoint).

\begin{lemma}
\label{lemma2}
Given a collection $S$ of $n$ strings of arbitrary lengths and a function $f$ from $S$ into $[0,m-1]$, we can build a data structure which uses $O(n\log m)$ bits and which which computes $f(x)$ for any $x\in S$ in time $O(|x|/w)$ (where $|x|$ is the length of $x$ in bits). When queried for any $y\notin S$ the function returns any value from the set $f(S)$. 
\end{lemma}
This result can easily be obtained using minimal perfect hashing~\cite{FKS84,HT01}. Though perfect hashing is usually defined for fixed $O(w)$ bits integers, a standard string hash function~\cite{DGMP} can be used to first reduce the strings to integers before constructing the minimal perfect hashing on the generated integers. 
%In the litterature the problem is called bloomier filter problem~\cite{CKRT04}, dictionary with retrieval only (~\cite{DP08}) or static function storage (~\cite{BBPV09}).  
%The reader can refer to one of (~\cite{DP08,CC08,P09}) for more complicated but more space-efficient solutions. 

\begin{lemma}\cite[Theorem~1]{CHSV08}
\label{lemma3a}
Given a collection $S$ of $n$ strings of variable lengths occupying a memory area of $m$ characters (the strings can possibly overlap), we can build an index which uses $O(n\log m)$ bits so that given any string $x$, we can find the string $s\in S$ which is the longest among all the strings of $S$ which are prefix of $x$ in time $O(|x|/w+\log n)$ (where $|x|$ is the length of $x$ in bits). More precisely, the data structure returns $prrank_S(s)$.
%the index in (prefix) lexicographic order of the string $s$ (relatively to the set $S$). 
Moreover the data structure is able to tell whether $x=s$.
\end{lemma}
This result which is obtained using a string B-tree~\cite{FG99} combined with an LCP array and a compacted trie~\cite{MM93} built on the set of strings, and setting the block size of the string B-tree to $O(1)$. The following lemma is symmetric of the previous one. 
\begin{lemma}
\label{lemma3b}
Given a collection $S$ of $n$ strings of variable lengths occupying a memory area of $m$ characters of space (the strings can possibly overlap), 
we can build an index which uses $O(n\log m)$ bits so that given any string $x$, we can find the string $s\in S$ which is the longest among all the strings of $S$ which are suffix of $x$ in time $O(|x|/w+\log n)$ (where $|x|$ is the length of $x$ in bits). More precisely, the data structure returns $surank_S(s)$.
%the index in (prefix) lexicographic order of the string $s$ (relatively to the set $S$). 
Moreover the data structure is able to tell whether $x=s$.
\end{lemma}

\begin{lemma}
\label{lemma4}\cite{C86}
Given a set of $n$ rectangles in the plane, we can build a data structure which uses $O(n\log n)$ bits of space so that given any point $[v,z]$, we can report all the $k$ occurrences of rectangles which enclose that point in time $O(\log n+k)$. 
\end{lemma}
The problem solved by lemma~\ref{lemma4} is called the $2D$ stabbing problem or sometimes called the planar point enclosure. The lemma uses the best linear space solution to the problem which is due to Chazelle~\cite{C86} (which is optimal according to the lower bound in ~\cite{P08a}). 

\section{Multiple string matching without tabulation}
\subsection{Overview}
The goal of this section is to show how we can simulate the running of the AC automaton~\cite{AC75}, by processing the characters of the scanned text in blocks of $b$ characters. The central idea of the relies on a reduction of the problem of dictionary matching to the $1D$ and $2D$ stabbing problems, in addition to the use of standard string data structures namely, string B-trees, suffix arrays and minimal perfect hashing on strings.
 %where ideally $b=\frac{w}{\log\sigma}$. 
At each step, we first read $b$ characters of the text, find the matching patterns which end at one of those characters and finally jump to the state which would have been reached after reading the $b$ characters by the AC automaton (thereby simulating all $next$ and $fail$ transitions which would have been traversed by the standard AC automaton for the $b$ characters). Finding the matching patterns is reduced to the $2D$ stabbing problems, while jumping to the next state is reduced to $1D$ stabbing problem. 
%of the text and then finding the next state to jump to is reduced to the $1D$ stabbing problem , while finding the matching patterns reduces to the $2D$ stabbing problem.
%The first problem can be solved in $O(\log\log m)$ using lemma~\ref{lemma1} built on a set of $m$ segments, while the second one could only be solved in $O(\log R)=O(\log d+\log b)$ time for a set of $R=db$ rectangles~\cite{C86}. 
%Both of the two bounds are essentially tight (see ~\cite{PT06} and ~\cite{P08a} for respective lower bounds), which means that our results can not be improved through improvements in solutions of the two problems. 
The geometric approach has already been used for dictionary matching problem and for text pattern matching algorithms in general. For example, it has been recently used in order to devise compressed indexes for substring matching~\cite{GV05,N04,CHSV08}. Even more recently the authors of~\cite{TWLY09} have presented a compressed index for dictionary matching which uses a reduction to $2D$ stabbing problem. 
%~\cite{AN07,CHSV08,HSTV09}. 
\subsection{The data structure}
%We now give the proof of theorem~\ref{theorem1}. 
We now describe the data structure for in more detail. Given the set $S$ of $d$ patterns, we note by $P$ the set of the prefixes of the patterns in $S$ (note that $|P|\leq m+1$). It is a well-known fact that there is a bijective relation between the set $P$ and the set of states of the AC automaton. We use the same state representation as the one used in~\cite{B10a}. That is we first sort the states of the automaton in the suffix-lexicographic order of the prefixes to which they correspond, attributing increasing numbers to the states from the interval $[0,m]$. Thus the state corresponding to the  empty string gets the number $0$, while the state corresponding to the greatest element of $P$ (in suffix-lexicographic order) gets the largest number which is at most $m$. 
%From now on, each state is implicitly represented by the number attributed to it. 
We define $state(p)$ as the state corresponding to the prefix $p\in P$. 
\\ 
Now, the characters of the scanned text, are to be scanned in blocks of $b$ characters. 
%More on this choice later. 
%The value $b$ is chosen as the minimum between $w\log\sigma$ (the number of characters which fit in a single memory word) and $y=min_{x\in S}|x|$ (the length of the shortest pattern). The reason for the choice is the following: Firstly, we clearly can not scan more than $\Theta(w)$ bits at once, as this is the maximum number of bits which can be read by the processor in a single step. Secondly, at each step we can not read more than $y=min_{x\in S}|x|$ as otherwise we could miss occurrences of the shortest patterns inside the pattern. We can remedy to that problem, but only through the use of tabulation. This is essentially what is done in the proof of theorem~\ref{theorem2}. Let's now describe more in detail how each step works.
For finding occurrences of the patterns in a text $T$, we do $\lceil n/b\rceil$ steps. At each step $i\in [0,\lceil n/b\rceil-1]$ we do three actions:
\begin{itemize}
\item Read $b$ characters of the text, $T[ib,(i+1)b-1]$ (or $n-ib\leq b$ characters of the text, $T[ib,n)$ in the last step).
\item Identify all the occurrence of patterns which end at a position $j$ of the text such that $j\in [ib,(i+1)b)$ ($j\in[ib,n)$ in the last step). 
\item If not in the last step go to the next state corresponding to the longest element of $P$ which is a suffix of $T[0,(i+1)b]$. 
\end{itemize}
%The details of the implementation of each of the last two actions is given in sections~\ref{subsec:ident_occ} and ~\ref{subsec:simul_trans}. 
The details of the implementation of each of the last two actions is given in sections ~\ref{subsec:ident_occ} and~\ref{subsec:simul_trans}. 
%Before that, we briefly mention the data structures which will be used as components of the data structure:
\\Our AC automaton representation has the following components:
\begin{enumerate}
\item An array $A$ which contains the concatenation of all of the patterns. This array clearly uses $mb$ bits of space.
\item Let $P_{0<i\leq b}$ be the set of prefixes of $S$ of lengths in $[1,b]$. We use an instance of lemma~\ref{lemma3a}, which we denote by $B_1$ and in which we store the set $P_{0<i\leq b}$ (by means of pointers into the array $A$). Clearly $B_1$ uses $O(db\log m)=O(m\log m)$ bits of space (we have $db$ elements stored in $B_1$ and pointers into $A$ take $\log m$ bits). We additionally store a vector of $|P_{0<i\leq b}|\leq db$ elements which we denote by $T_1$ and which associates an integer in $[0,m)$ with each element stored in $B_1$. The table $T_1$ uses $O(db\log m)=O(m\log m)$
%This data structure uses at most $O(db\log m)$ bits with each element. 
%The content of $T_1$ which uses also $O(db\log m)=O(m\log m)$, is defined more in detail in appendix~\ref{appendix:simul_trans}.

\item We use an instance of lemma~\ref{lemma2}, which we denote by $B_2$ and in which we store all the suffixes of strings in $P$ (or equivalently all factors of the strings in $S$) of length $b$ and for each suffix, store a pointer to its ending position in the array $A$ (if the same factor occurs multiple times in the $S$ we store it only once). As we have at most $m$ elements in $P$ and each pointer (in the array $A$) to each factor can be encoded using $O(\log m)$ bits, we conclude that $B_2$ uses at most $O(m\log m)$ bits of space. 
\item We use an instance of lemma~\ref{lemma3b} which we denote by $B_3$ and in which we store all the suffixes of strings of $S$ of lengths in $[1,b]$ (We note that set by $U_{0<i\leq b}$). It can easily be seen that $B_3$ also uses $O(db\log m)=O(m\log m)$ bits of space. 
\item We use a $1D$ stabbing data structure (lemma~\ref{lemma1}) in which we store $m$ segments where each segment corresponds to a state of the automaton. This data structure which uses $O(m\log m)$ bits of space is used in order to simulate the transitions in the AC automaton. 
%More details on the use and content of the data structure are given in appendix~\ref{appendix:simul_trans}. 
We also store a vector of integers of size $m$ which we denote by $T_2$ and which associates an integer with each interval stored in the $1D$ stabbing data structure. The table $T_2$ uses $O(m\log m)$ bits of space.
%The content of table $T_2$  which uses $O(m\log m)$ bits is defined more in detail in appendix~\ref{appendix:simul_trans}. 
\item We use a $2D$ stabbing data structure (lemma~\ref{lemma4}) in which we store up to $db$ rectangles. The space used by this data structure is $O(db\log (db))=O(m\log m)$ bits. We also use a table $T_3$ which stores triplets of integers associated with each rectangle. The table $T_3$ will also use $O(db\log m)=O(m\log m)$ bits. 
%More details on the use and contents of the $1D$, $2D$ stabbing data structures and the table $T_1$, $T_2$ and $T_3$ are given in appendix~\ref{appendix:ident_occ}.

\end{enumerate}
We deffer the details about the contents of each component to the full version which uses also to the full version.
Central to the working of our data structure is the following technical lemma:
\begin{lemma}
\label{lemma:tech_lemma}
Given a set of strings $X$. We have that for any two strings $x\in X$ and $y\in X$:
\begin{itemize}
\item $prrank_X(y)\in [prrank_X(x),prrank_X(x)+prcount_X(x)-1]$ iff $x$ is a prefix of $y$. 
\item $surank_X(y)\in [surank_X(x),surank_X(x)+sucount_X(x)-1]$ iff $x$ is a suffix of $y$.
\end{itemize}
\end{lemma}
The proof of the lemma is omitted.
\subsection{Simulating transitions}
\label{subsec:simul_trans}
We will use the representation of states similar to the one used in~\cite{B10a}. That is each state of the automaton corresponds to a prefix $p\in P$ and is represented as an integer $state(p)=surank_P(p)$. The main idea for accelerating transitions is to read the text into blocks of size $b$ characters and then find the next destination state attained after reading those $b$ characters using $B_1$, $T_1$, $B_2$, $T_2$ and the $1D$ stabbing data structure. More precisely being at a state $state(p)$ and after reading next $b$ characters of the text which form a string $q$, we have to find next state which is the state $state(x)$ such that $x\in P$ is the longest element of $P$ which is suffix of $pq$. For that purpose the $1D$ stabbing data structure is used in combination with $B_1$ (which is queried on string $q$) in order to find $state(x)$ in case $|x|\geq b$. Otherwise if no such $x$ is found the data structure $B_2$ will be used to find $state(x)$, where $|x|<b$. 
The following lemma summarizes the time and the space of the data structures needed to simulate a transition.  
\begin{lemma}
\label{lemma:trans_lemma}
We can build a data structure occupying $O(m\log m)$ bits of space such that if the automaton is in a state $t_i$, the state $t_{i+b}$ reached after doing all the transitions on $b$ characters, can be computed in $O(\log d+\log b+\log\log m+\frac{b\log\sigma}{w})$ time. 
\end{lemma}
%The details of the data structure and the proof of the lemma are deferred to appendix~\ref{appendix:simul_trans}.
%The proof of the lemma is deferred to the full version.
%As we said previously 
The current state of the AC automaton is actually represented as a value $cur\in[0,m]$. At the beginning the automaton is at state $cur=0$, and we read the text in blocks of $b$ characters at each step. At the end of each step we have to determine the next state reached by the automaton which is represented by the number $next\in[0,m]$. We now show how the transitions of the AC automaton for a block of $b$ characters are simulated. Suppose that we are at step $i$ and the automaton is in the state $state(p)$ corresponding to a prefix $p$. Now we have to read the substring $q=T[ib,(i+1)b-1]$ 
%(or $q=T[ib,n-1]$ in the last step $i=\lceil n/b\rceil-1$). 
and the next state to jump to after reading $q$, is the state $state(x)$ corresponding to the longest element $x$ of $P$ which is a suffix of $pq$. 
%We can decompose $x$ into two strings $p=p'q'$, where $q'=q$ is a suffix of $p$ of length $b$ and $p'$ is a prefix of $p$ of length $|p|-b$. 
\\
For simulating transitions we use $B_1$,$T_1$, $B_2$, $T_2$ and the $1D$ stabbing data structure. The table $T_1$ associates to each element of $P_{0<i\leq b}$ (each element of $P$ whose length is in $[1,b]$) sorted in suffix lexicographic order the identifier of the states to which they correspond. That is for each $x\in P_{0<i\leq b}$ we set $T_1[surank_{P_{0<i\leq b}}(x)]=state(x)$. We recall that given any element $x\in P_{0<i\leq b}$, $surank_{P_{0<i\leq b}}(x)$ can be obtained by querying $B_1$ for the element $x$.  
% That is the cell number $j$ in the table stores the state associated with the prefix of $P_{0<i\leq b}$ of rank $j$ (in suffix lexicographic order relatively to the set $P_{0<i\leq b}$). 
\\
The $1D$ stabbing data structure (lemma~\ref{lemma1}) which is built on numbers occupying $2\log (m+1)$ bits each, stores $m$ intervals each of which is defined by two points, where each point is defined by a number which occupies $2\log (m+1)$ bits. 
%The query time of the stabbing data structure 
%(which in fact is a predecessor data structure)
% is $O(\log\log m)$. 
Let $x\in P$ be decomposed by $x=p'q'$ where $q'$ is the suffix of $x$ of length $b$ and $p'$ is the prefix of $x$ of length $|x|-b$. Let $ID(q')$ be the pointer associated with $q'$ in $B_2$ (recall that $B_2$ associates a unique pointer in $A$ for each occurring factor $q'$ of elements of $S$). We store in the $1D$ stabbing data structure the interval $[I_0,I_1]$, where $I_0=ID(q')\cdot state(p')$ and $I_1=ID(q')\cdot (state(p')+sucount_P(p')-1)$ (recall that $sucount_P(p')$ is the number of elements of $P$ which have $p'$ as a suffix). The $1D$ stabbing data structure naturally associates a unique integer identifier from $[0,m]$ with each interval stored in it. We additionally use a table $T_2$ of size $m$ indexed with the interval identifiers. More precisely, let $j$ be the identifier corresponding to the interval associated with the state $state(p)$ for $p\in P$. We let $T_2[j]=state(p)$. That  way once we have found a given interval from the $1D$ stabbing data structure, we can index into table $T_2$ in order to find the corresponding state. 
\\Now queries will happen in the following way: At step $i$, we are at state $state(p)$ corresponding to a prefix $p$ and we are to read the sequence $q=T[ib,(i+1)b-1]$, and must find the longest element of $P$ which is a prefix of $pq$. For that, we do the following steps:
\begin{enumerate}
\item We first query $B_2$ for the string $q$ which will return a unique identifier $ID(q')$ which is in fact a pointer to the ending position of a factor $q'$. Now, we compare $q'$ with $q$. If they are not equal, we go to step $5$, otherwise we continue with the next step. 
\item We query the $1D$ stabbing data structure for the point $ID(q)\cdot state(p)$ This query returns the interval (identified by a variable $j$) which most tightly encloses the point $ID(q)\cdot state(p)$ if it exists. This interval (if it exists) corresponds to a prefix $x=p'q'$ of $P$ such that $q'=q$ and $p'$ is the longest element of $P$ which is a prefix of $p$. If the query returns no interval, we conclude that we have no element of $P$ of length $\geq b$ which is a suffix of $pq$ and go to step $5$, otherwise we continue with the next step.
\item We retrieve $T_2[j]$ which gives us the destination state which concludes the transition. 
\item At this step we are sure that no element of $P$ of length at least $b$ is a suffix of $pq$. We thus do a query on $B_1$ for the string $q$ in order to find the longest element of $P$ which is a suffix of $pq$. Note that this element must be of length $<b$ and thus must be stored in $B_1$ and also must be a suffix of $q$. Let $ID(q)$ be the identifier of the returned element. 
\item By reading $T_1$ we retrieve the identifier of the destination state which is given by $T_1[ID(q)]$. This concludes the transition. 
\end{enumerate}
We now give a formal proof of lemma~\ref{lemma:trans_lemma}

\begin{proof} 
We now prove that the above algorithm effectively simulates $b$ consecutive transitions in the automaton. Recall that we are looking for the state corresponding to the longest element $x\in P$ which is a suffix of $pq$. 
After we have read the string $q$, we query the data structure $B_2$ to retrieve a pointer $ID(q')$ to a string $q'$ which is a factor of some string in $S$ (or equivalently a suffix of some element in $P$). Then we compare $q$ with $q'$ in time $O(b\log\sigma/w)$. Now we have two cases: 
\begin{itemize}
\item The comparison is not successful, we conclude that no prefix in $P$ has $q$ as a suffix and hence the element $x\in P$ must be shorter than $b$ (otherwise it would have had $q$ as a suffix). That means that $x$ is a suffix of $q$ ($x$ is a suffix of $pq$ shorter than $q$) and hence has length at most $b$. Hence we go the step $4$ to query $B_1$ for the string $q$ in order to retrieve $x$.
\item The comparison is successful, in which case we know there exists at least one element of $P$, which has $q$ as a suffix. Now we go to step $2$ , querying the $1D$ stabbing for the point $K=ID(q)\cdot state(p)$. The query returns an interval $[I_0,I_1]$ where $I_0=ID(q')\cdot state(p')$ and $I_1=ID(q')\cdot (state(p')+sucount_P(p')-1)$ for some prefixes $p'$ and $q'$. Now it can easily be proven that $q'=q$ and that $p'q$ is the longest element of $P$ which is a suffix of $pq$. This is proved by contradiction. By lemma~\ref{lemma:tech_lemma} we have that $p'$ must be a suffix of $p$, and we suppose that the longest suffix is $p''\neq p'$ having an associated interval $[J_0,J_1]$ and $K\in [J_0,J_1]$. By definition $p'$ is a suffix of $p''$  and thus by lemma~\ref{lemma:tech_lemma} $[J_0,J_1]\subseteq [I_0,I_1]$ which contradicts the fact that $[I_0,I_1]$ is among all the intervals stored in the $1D$ data structure the one which most tightly encloses $K$ (which is implied by lemma~\ref{lemma1}). 
\end{itemize}
Now, in the first case, we go to step $4$ in order to find the longest prefix in $P$ which is a suffix $q$. In the second case, we go to step $2$ looking among the elements which have $q$ as a suffix for the longest one which is a suffix of $pq$. If the search is unsuccessful, we conclude that no such element $x$ exists and thus $x$ must be shorter than $q$ and thus go to step $4$ to find the longest prefix in $P$ which is a suffix of $q$. 
\\
The total space usage is clearly $O(m\log m)$ bits as each of $B_1$,$B_2$,$T_1$,$T_2$ and the $1D$ stabbing data structure uses $O(m\log m)$ bits. 
\\
Concerning the query time, it can easily be seen that the steps $3$ and $5$ take constant time, step $1$ takes time $O(\frac{b\log\sigma}{w})$, step $2$ takes time $O(\log\log m)$ and finally step $4$ takes $O(\frac{b\log\sigma}{w}+\log d+\log b)$ time. Summing up, the total time for a transition is $O(\frac{b\log\sigma}{w}+\log d+\log b+\log\log m)$.
\end{proof}

\subsection{Identifying matching occurrences}
\label{subsec:ident_occ}
In order to identify matching patterns the $2D$ stabbing data structure is used in combination with $B_1$. 
%The details are deferred to appendix~\ref{appendix:ident_occ}.
\begin{lemma}
\label{lemma:report1_lemma}
Given a parameter $b$ and a set $S$ of variable length strings of total length $m$ characters over an alphabet of size $\sigma$, we can build a data structure occupying space $O(m \log m)$ bits, such that if the automaton is at a state $t_i$ after reading $i$ characters of a text $T$, all the $occ_i$ matching occurrences of $T$ which end at any position in $T[i,i+b]$ (or $T[i,|T|-1]$ if $i+b\geq |T|$) and begin at any position in $T[0,i]$ can be computed in $O(\log d+\log b+\frac{b\log\sigma}{w}+occ_i)$ time. 
\end{lemma}
In order to find the matching pattern occurrences at each step, we use $B_3$, the table $T_3$ and the $2D$ stabbing data structure.
Initially the automaton is at state $0$, we read the first $b$ characters of the text, $T[0,b)$ and must recognize all occurrences which end in any position $j\in[0,b)$. Note that in this first step, any occurrence must end at position $b-1$ (This is the case, because we have assumed that $b$ is no longer than the length of the shortest pattern). Then at each subsequent step $i$, we read a block $T[ib,(i+1)b)$ (or the block $T[ib,n)$ in the last step) and must recognize all the occurrences which end at position $j\in[ib,(i+1)b)$ (or $j\in[ib,(i+1)b)$ in the last step). Suppose that at some step $i$ we are at a state $state(p)$ corresponding to a prefix $p$ and we are to read the block $q=T[ib,(i+1)b)$ ($q=T[ib,n)$ in the last step). It is clear that any matching occurrence must be a substring of $pq$ (the string $p$ concatenated to the string $q$) and moreover, that substring must end inside the string $q$. In other words, any occurrence $x$ is such that $x=p'q'$, where $p'$ is suffix of $p$ and $q'$ is prefix of $q$. 
\\
Identifying the pattern involves first computing a point $[x_p,y_q]$ where $x_p=state(p)$ and $y_q=prrank_{U_{0<i\leq b}}(q)$ is computed by querying $B_3$ for $q$, then querying the $2D$ stabbing data structure in order to get all the rectangles which enclose $[x_p,y_q]$ as integer identifiers, where each reported rectangle represents one occurrence of one of the patterns. Finally using the table $T_3$, we can get the matching pattern identifiers along their starting and ending positions. 
We now describe how the set of rectangles is built. For each pattern $s\in S$ of length $l$ we insert $b$ rectangles. Namely, for each $i\in[1,b]$ we insert the rectangle which is defined by the two intervals:
\begin{itemize}
\item Let $p'$ be the prefix of $s$ of length $l-i$. Let $R=state(p')=surank_P(p')$ be the state corresponding to $p'$ (or equivalently the rank of $p'$ in suffix-lexicographic order relatively to the set $P$) and let $c=sucount_{P}(q')$ be the number of elements of $P$ which have the string $p'$ as a suffix. The first interval is given by $[R,R+c-1]$.
\item Let $q'$ be the suffix of $s$ of length $i$. Let $ID(q')=prrank_{U_{0<i\leq b}}(q')$ be the unique identifier returned by $B_3$ for $q'$ (recall that $B_3$ stores all suffixes of lengths at most $b$ of elements of $P$). Let $c=prcount_{U_{0<i\leq b}}(q')$ be the number of elements of $U_{0<i\leq b}$ which have $q'$ as a prefix. The second interval is given by $[ID(q'),ID(q')+c-1]$. 
\end{itemize}
The $2D$ stabbing data structure returns a unique identifier $j\in[0,db-1]$ corresponding to each rectangle. Additionally with the rectangle, we associate a triplet $(I,|p'|,|q'|)$ which is stored in table $T_3$ at position $T3[j]$, where $I\in[0,d-1]$ is the unique integer identifier of the pattern $s$. This table thus uses $O(db\log m)=O(m\log m)$ bits of space.
\\Now queries will happen in the following way: suppose that we are at state $state(p)$ corresponding to a prefix $p$ and we are to read the block $q=T[ib,(i+1)b)$. We first query $B_3$ for the string $q$ giving us an identifier $ID(q')=prrank_{U_{0<i\leq b}}(q')$ corresponding to the longest element $q'\in U_{0<i\leq b}$ such that $q'$ is prefix of $q$. Then we do a $2D$ stabbing query for the point $(state(p),ID(q'))$. Now for every found rectangle identified by an integer $j$, we retrieve the triplet $(I,|p'|,|q'|)$ from $T_3[j]$. Now the reported string has identifier $I$, and matches the text at positions $[ib-|p'|,ib+|q'|-1]$.  
We now give a formal proof of lemma~\ref{lemma:report1_lemma}

\begin{proof}
We now prove that the above procedure reports all (and only) matching occurrences. For that it suffices to prove that there exists a bijection between occurrence and reported rectangles. It is easy to see that each occurrence $s$ which begins in $T[0,i]$ and ends in $T[i,i+b]$ can be decomposed as $s=p'q'$, where $p'$ is a suffix of $T[0,i-1]$ ($p'$ can possibly be the empty string) and $q'$ is a prefix of $T[i,i+b]$. Then as $s\in S$, we can easily deduce that $q'\in U_{0<i\leq b}$ and $p'\in P$. It is also easy to see that $p'$ is a suffix of $p$. Let $q\in U_{0<i\leq b}$ be the longest element in $U_{0<i\leq b}$ which is a prefix of $T[i,i+b]$. It is easy to see that $q'$ must be a prefix of $q$. Thus according to lemma~\ref{lemma:tech_lemma} we have that $surank_P(p)\in [surank_P(p'),surank_P(p')+sucount_P(p')-1]$ and $prrank_{U_{0<i\leq b}}(q)\in [prrank_{U_{0<i\leq b}}(q'),prrank_{U_{0<i\leq b}}(q')+prcount_{U_{0<i\leq b}}(q')-1]$. Now recall that $B_3$ returns $prrank_{U_{0<i\leq b}}(q)$ and the $2D$ stabbing query is done precisely on the point $[state(p),prrank_{U_{0<i\leq b}}(q)]$, which will thus return all (and only) the rectangles corresponding to occurrences. 
%We now prove that the stabbing query reports all the rectangles which correspond to any element $f\in S$, such that $f=p'q'$ and $p'$ (respectively $q'$) is a suffix (prefix) of $p$ ($q$). It can easily be seen that the above data structure works correctly. Notice that $B_3$ returns the unique identifier $ID(q')=prrank_{U_{0<i\leq b}}(q')$ corresponding to the longest element $q'$ stored in $B_3$ which is a prefix of $q=T[ib,j]$. Likewise, we know that $p$ is the longest element of $P$ which is a suffix of $T[0,ib)$. Now it can easily be seen that any factor $f\in S$ of $T[0,j]$ of length at least $b$ and which ends inside $q=T[ib,j]$ can be decomposed into $p_fq_f$, where $p_f\in P$ is a suffix of $p$ and $q_f\in U_{0<i\leq b}$ is a prefix of $q$. 
\end{proof}
By combining lemma~\ref{lemma:report1_lemma} and lemma~\ref{lemma:trans_lemma}, we directly get theorem~\ref{theorem1} by setting $b=y$ where $y$ is the length of the shortest string in the set $S$. At any step, we do the following: 
\begin{enumerate}
\item Read the characters $T[i,j]$, where we set $i=Ib$ and we set $j=(I+1)b-1$ if $n>(I+1)b-1$ and $j=n-1$ otherwise. 
\item Recognize all the pattern occurrences which start at positions any position $i'\leq i$ and which terminate at positions in $[i,j]$ using lemma~\ref{lemma:report1_lemma}. 
\item Increment step $I$ by setting $I=I+1$. Then if $Ib<n$, stop the algorithm immediately. 
\item Do a transition using lemma~\ref{lemma:trans_lemma} and return to action 1. 
\end{enumerate}
\subsection{Analysis}
Theorem~\ref{theorem1} is obtained by combining lemma~\ref{lemma:trans_lemma} with lemma~\ref{lemma:report1_lemma}. Namely by setting $b=y$, where $y$ is the shortest pattern in $S$ in both lemmata we can simulate the running of the automaton in $\lceil n/y\rceil$ steps at each step i, spending $O(\log d+\log b+\log\log m+\frac{y\log\sigma}{w})+occ_i)$ to find the $occ_i$ matching occurrences (through lemma~\ref{lemma:report1_lemma}) and $O(\log d+\log b+\log\log m+\frac{b\log\sigma}{w})$ time to simulate the transitions (through lemma~\ref{lemma:trans_lemma}). Summing up over all the $\lceil n/y\rceil$ steps, we get the query time stated in the theorem.
We can now formally analyze the correctness and space usage of theorem~\ref{theorem1}.  
\paragraph{Correctness}
The correctness of the query is immediate. If can easily be seen that at each step $I$, we are recognizing all occurrences which end at any position in $[Ib,I(I+1)-1]$ (or $[Ib,n-1]$ in the last step). That is at any step $I$ we can use lemma~\ref{lemma:report1_lemma} to recognize all the occurrence which end at any position in $[Ib,(I+1)b-1]$ (or $[Ib,n-1]$ in the last step) and start at any position $i'\leq Ib$. Also at each step $I$, we are at state $st_{Ib}$ reached after reading $Ib$ characters of the text, and lemma~\ref{lemma:trans_lemma} permits us to jump to the state $st_{(I+1)b}$ which is reached after reading $b$ characters.
\paragraph{Space usage}
Summing up the total space usage by the theorem is $O(m\log m)$ bits as both lemma~\ref{lemma:trans_lemma} and lemma~\ref{lemma:report1_lemma} use $O(m\log m)$ bits.

\subsection{Consequences}
Theorem~\ref{theorem1} states that we can use $O(m\log m)$ bits of space to identify all occurrence of length at least $y$ in a text $T$ of length $n$ in time $O(n(\frac{\log d+\log y+\log\log m}{y}+\frac{\log\sigma}{w})+occ)$. If we suppose that all the patterns are of length at least $w$ bits ($\frac{w}{\log\sigma}$ characters), then by setting $y=\frac{w}{\log\sigma}$, we obtain an index which answers to queries in time $O(n(\frac{\log\sigma(\log d+\log\frac{w}{\log\sigma}+\log\log m)}{w}+\frac{\log\sigma}{w})+occ)$. As we have $\log m<w$ and $\log\frac{w}{\log\sigma}<\log w$, the query time simplifies to  $O(n\frac{\log\sigma(\log d+\log\log w)}{w}+occ)$. This gives us corollary~\ref{corollary2}. 
%\\In the case of single string matching where we set $d=1$ and $y=m$ and thus the query time becomes $O(n(\frac{\log m}{m}+\frac{\log\sigma}{w})+occ)$. This gives us corollary~\ref{corollary1}. 
%\\
An important implication of theorem~\ref{theorem1} is that it is possible to attain the optimal $O(n\frac{\log \sigma}{w}+occ)$ query time in case the patterns are of sufficient minimal length. Namely if each pattern is of length at least $(\log d+\log w)$ words (that is $\frac{w(\log d+\log w)}{\log\sigma}$ characters), then by setting $y=\frac{w(\log d+\log w)}{\log\sigma}$ in theorem~\ref{theorem1}, we obtain a query time of $O(n\frac{\log \sigma}{w}+occ)$.  This gives us corollary~\ref{corollary3}.  
\section{Tabulation based solution for multiple-string matching}
\label{section:tabul_sol}
We now prove theorem~\ref{theorem2}. 
A shortcoming of theorem~\ref{theorem1} is that it gives no speedup in case the length of the shortest string in $S$ is too short. In this case we resort to tabulation in order to accelerate matching of short patterns. More specifically, in case, we have a specified quantity $t$ of available memory space (where $t<2^w$ as obviously we can not address more than $2^w$ words of memory), we can precompute lookup tables using a standard technique known as the four russian technique~\cite{ADKF70} so that we can handle queries in time $O(n\frac{\log d+\log\log_\sigma t+\log\log m}{\log_\sigma t}+occ)$. In theorem~\ref{theorem1} our algorithm reads the text in blocks of size $b=y$, where $y$ is the length of the shortest pattern. In reality we can not afford to read more than $y$ characters at the each step, because by doing so we may miss a substring of the block of length $y$. Thus in order to be able to choose a larger block size $b$, we must be able to efficiently identify all substrings of any block of (at most) $b$ characters which belong to $S$. The idea is then to use tabulation to answer to such queries in constant times (or rather in time linear in the number of reported occurrences). More in detail, for each possible block of $u\leq b$ characters, we have a total of $(u-1)(u-2)/2$ substrings which could begin at all but the first position of the block. For each possible block of $u$ characters, we could store a list of all substrings belonging to $S$ and each list takes at most $(u-1)(u-2)/2=O(u^2)$ pointers of length $\log m$ bits. As we have a total of $\sigma^u$ possible characters, we can use a precomputed table of total size $t=O((\sigma^u)u^2\log m)$ bits. %That is for each possible block, we run a dictionary matching algorithm for all strings of length less than $b$ and store in the list all occurrences which begin at all but the first position in the block. 
%Thus, we will have to choose the maximal $b$ such that $\sigma^b(b-1)(b-2)/2\leq t$. Clearly we have $b=\Theta(\log_\sigma t)$. 
%The corollary~\ref{corollary4} follows from theorem~\ref{theorem2} by setting $t=n^{\epsilon}$. 
%The construction of the tables will then take time $O(n^\epsilon)\log^2_\sigma n=o(n)$.  
\begin{lemma}
\label{lemma:report2_lemma}
For a parameter $u\leq \epsilon w/\log\sigma$ (where $\epsilon$ is any constant such that $0<\epsilon<1$) and a set $S$ of patterns where each pattern is of length at most $u$, we can build a data structure occupying $O(\sigma^u\log^2 u\log m)$ bits of space such that given any string $T$ of length $u$, we can report all the $occ$ occurrences of patterns of $S$ in $T$ in $O(occ)$ time. 
\end{lemma}
Theorem~\ref{theorem2} is obtained by combining lemmata~\ref{lemma:trans_lemma},~\ref{lemma:report1_lemma} and ~\ref{lemma:report2_lemma}.
Suppose we are given the parameter $\alpha$; for implementing transitions, we can just use lemma~\ref{lemma:trans_lemma} in which we set $b=s$, where the transitions are built on the set containing all the patterns. Now in order to report all the matching strings, we build an instance of lemma~\ref{lemma:report1_lemma} on the set $S$ and in which we set $b=s$ and also build $\alpha-1$ instances of lemma~\ref{lemma:report2_lemma} for every $u$ such that $1\leq u<\alpha$. More precisely let $S_{\leq u}$ be the subset of strings in $S$ of length at most $u$, then the instance number $u$ will be built on the set $S_{\leq u}$ using parameter $u$ and will thus for all possible strings of length $u$, store all matching patterns in $S$ of length at most $u$.  
\\ A query on a text of $T$ will work in the following way: we begin at step $I=0$ and the automaton is at state $0$ which corresponds to the empty string. Recognizing the patterns will consist in the following actions done at each step $I$: 
\begin{enumerate}
\item Read the substring $T[i,j]$, where $i=Ib$ and $j=(I+1)b-1$ (or $j=n-1$ if $n>(I+1)b-1$). 
\item Recognize all the pattern occurrences which start at any position $i'\leq i$ and which terminate at any position $j'\in [i,j]$ using lemma~\ref{lemma:report1_lemma}. 
\item Recognize all the matching strings of lengths at most $b$ which are substrings of $T[i+1,j]$ using the instance number $j-i$ of lemma~\ref{lemma:report2_lemma}.
\item Increment step $I$ by setting $I=I+1$. Then if $Ib>n$, stop the algorithm immediately. 
\item Do a transition using lemma~\ref{lemma:trans_lemma} and return to action 1. 
\end{enumerate}
\subsection{Analysis}
\paragraph{Correctness}
The correctness of the transition is immediate. If can easily be seen that at each step $I$, we are recognizing all occurrences which end at any position in $[Ib,(I+1)b-1]$ (or $[Ib,n-1]$ in the last step). That is at any step $I$:
\begin{itemize}
\item Lemma~\ref{lemma:report1_lemma} recognizes all occurrence which end at any position in $[Ib,(I+1)b-1]$ and start at any position $i'\leq Ib$.
\item Lemma~\ref{lemma:report2_lemma} recognizes all patterns which end at any position in $[Ib+1,(I+1)b-1]$ and start at any position $i'$ such that $i'\in[Ib+1,(I+1)b-1]$ using instance number $\alpha-1$ of the lemma for all but the last step. 
\item The last step of the algorithm recognizes all occurrences which end at any position in $[Ib+1,n-1]$ and start at any position $i'\in[Ib+1,n-1]$ using the instance number $u=n-Ib-1$ of lemma~\ref{lemma:report2_lemma}. 
\end{itemize} 
Thus at the last step, we will have recognized all the occurrences of patterns in the text $T$. 

\paragraph{Space usage}
It can easily be seen that the space used by lemma~\ref{lemma:trans_lemma} is in fact $O(m\log m)$. The space used by lemma~\ref{lemma:report1_lemma} is $O(d\alpha\log (d\alpha))=O(m\log m)$ bits. The space used by all the instances of lemma~\ref{lemma:report2_lemma} is bounded above by $O(\sigma^\alpha\log^2\alpha\log m)$. That is for each $u\in[1,\alpha-1]$, the instance number $u$ uses $O(\sigma^u\log^2 u\log m)\leq c(\sigma^u\log^2 u\log m)$ for some constant $c$. Thus the total space usage is upper bounded by $$c(\sum_{u=1}^{s-1}\sigma^u\log^2 u\log m)$$ 
As we have $\sigma\geq 2$, the total space used by lemma~\ref{lemma:report2_lemma} can be upper bounded by $2c(\sigma^{\alpha-1}\log^2 (\alpha-1)\log m)=O(\sigma^s\log^2 \alpha\log m)$. Summing up the space used by the three lemmata we get $O(m\log m+\sigma^\alpha\log^2\alpha\log m)$ bits of space.

\subsection{Consequences}
Corollary~\ref{corollary4} derives easily from the theorem. That is, in the case $n\geq m$, we can set $\alpha=c\log_\sigma n$ for some constant $c<\epsilon/2$ for any $\epsilon\in(0,1)$. Then space usage becomes $O(m\log m+\sigma^\alpha\log^2 \alpha\log m)=O(m\log m+n^c(\log\log_\sigma n)^2\log m))=O(m\log m+n^{\epsilon})$.
\\
Similarly corollary~\ref{corollary5} derives immediately from the theorem. By setting $\alpha=c\log_\sigma m$ for some constant $c<1$, the space usage becomes $O(m\log m+\sigma^\alpha\log^2 \alpha\log m)=O(m\log m+m^c(\log\log_\sigma m)^2\log m))=O(m\log m)$ bits of space.

\section{Single string matching}
We now turn to the proof of theorems~\ref{theorem3} and ~\ref{theorem4}. In both theorem we only have to match a single pattern $p$ of length $m$ against the text $T$ of length $n$. We first describe the matching algorithm used in theorem~\ref{theorem3} then sketch a possible way to construct the data structure used in the matching algorithm. We finally show the proof of~\ref{theorem4} which is based on the data structure of theorem~\ref{theorem3} combined with the use of the four russian technique.
\subsection{The matching algorithm}
Our string matching algorithm employs properties of periodic strings. For implementing the string matching we will use a sliding window of size $m+h$ where $h=\lfloor m/3\rfloor$ and at each step $i$, we shift the window by $h+1$ characters and spend time $O(m\frac{\log\sigma}{w}+1)$. Thus the total running of the string matching will clearly be $O(\frac{n}{h}(m\frac{\log\sigma}{w}+1))=O(n(\frac{1}{m}+\frac{\log\sigma}{w}))$. We note by $P$ the set of all factors of the string $p$ of length $m-h$. Notice that $|P|\leq h+1$ (possibly the same factor  could occur multiple time in the string $p$ ). 
%Now we have two cases, either the pattern is periodic or aperiodic, and will treat each case separately.
%In case, the pattern is aperiodic, we will have the following. 
At any step $i$ and we consider a text windows $W=T[i(h+1),(i+1)(h+1)+m]$ and match every string which appears in the window. In other words match every string which starts at any position in $W[0,h]=T[i(h+1),(i+1)(h+1)-1]$ will be matched. Now let $q=W[h,m-1]$ (note that $|q|=m-h=\lceil 2m/3\rceil$). We will match $q$ against all factors of $p$ of length $m-h$. Note that for a match to be possible we must necessarily find that $q\in P$ (if $q$ is not a factor of $p$ then we can not have a match). Every occurrence which begins at any position in $W[0,h]$ must end at a position in $W[m-1,h+m-1]$. 
If the pattern $q$ occurs a single time in the windows $W$, then we just have to do a single comparison which takes optimal time. Thus from now on we concentrate on the case where the factor $q$ occurs multiple times in the pattern $p$. Before detailing the way the matching is done, we first prove the following lemma:
\begin{lemma}
\label{lemma:substr_count}
If $q$ appears $i$ times in $W$ and $i'$ times in $p$, then $p$ occurs at most $i-i'+1$ times in $W$
\end{lemma}
\begin{proof}
Let $q_1,q_2,,,,q_i$ be the sequence of consecutive appearances of $q$ in $W$. Then any occurrence of $p$ in $W$ must span exactly a sequence of $i'$ consecutive occurrences. As we exactly have $i-i'+1$ sequences of length $i'$ in a sequence of length $i$, we deduce that we have at most $i-i'+1$ occurrence of $p$ in $W$. 
\end{proof}
Thus our first step of the matching will be to count the number of occurrences of $q$ in $W$ which gives us an upper bound on the number of occurrence of $p$ in $W$. 
\paragraph{Counting occurrences of $q$ in $W$}
First note that in case the factor $q$ appears more than one time in $p$, then its (shortest) period is necessarily shorter than $|q|=m-h\geq |q|>2$   and can be uniquely decomposed as $q=(uv)^tu$ where $t\geq 2$, $|v|\geq 1$ and $|uv|$ is the length of the (shortest) period of $q$~\footnote{$q$ in this case (it hash a period of length less than $|q|$<2) is said to be periodic}. This can easily be explained: as $|q|\geq 2|p|/3$ (that is we have $|q|=m-h=\lceil 2m/3\rceil\geq 2m/3$) then necessarily any two occurrences of $q$ in $p$ are separated by at most $m-h\leq |p|/3\geq |q|/2$ characters. That means that $q$  has one period of length at most $|q|/2$ and thus the (shortest) period of $q$ is at most $|q|/2$. By a well known result (see Crochemore et al's book~\cite[Lemma 1.6]{CHL07}) we deduce that all periods of length at most $|q|/2$ of $q$ are multiple of the (shortest) period.
We now describe the way the number of occurrences of $q$ in $W$ are counted. Let $q'=W[0,h-1]$, $q''=W[m,h+m-1]$ (that is $W=q'qq''$) and $g=|uv|$. For counting we first do a longest suffix repetition search for $uv$ in $q'$ and then do a longest prefix repetition search of $vu$ in $q''$ returning two numbers $i'$ and $i''$ respectively. We now deduce that we have exactly $c_{q,W}=i'+i''+1$ occurrences of $q$ in $W$:
\begin{lemma}
The algorithm above correctly computes the number of occurrences of $q$ in $W$. 
\end{lemma}
\begin{proof}
The string $W$ contains a substring $s=(uv)^{i'+t+i''}u$. This substring contains at least $i+i''+1$ occurrences of $q=(uv)^tu$. We store a perfect hashing function so that for each factor $q$ we associate a triplet $(\alpha_q,\beta_q,r_q)$, where $\alpha_q$ is a pointer to first occurrence of the factor, $\beta_q$ is number of occurrences of $q$ and $r_q$ is the period of $q$ (emphasize that we need $r_q$ only in case $\beta_q\geq 2$). 
What remains is to prove that $W$ contains no more than $i'+i''+1$ occurrences of $q$. The proof is by contradiction. Suppose that there is an occurrence $W[\alpha,\alpha+m-1]=q$ which was outside of the substring $s=(uv)^{i'+t+i''}u$. We have two cases:
\begin{itemize}
\item The occurrence $W[\alpha,\alpha+m-1]$ is at the left of the occurrence $W[h,m-1]$, that is $\alpha<h$. In this case notice that $h-\alpha<h\leq|q/2|$. We have $q=W[\alpha,\alpha+m-1]=W[h,m-1]$ meaning that $q[t]=W[\alpha+t]=W[h+a]=q[h-\alpha+a]$ for any $a$ such that $h+t\leq |q|$. This means that $h-\alpha<|q|/2$ is a period of $q$ which moreover must be multiple of the shortest period $uv$. From there we deduce that $q[0,h-\alpha]=W[\alpha,h]=(uv)^b$ for some integer $b$ which means that the string $W[\alpha,h]$ will be included in the substring $uv)^{i'}$ by means of the longest suffix matching. Hence we conclude that the occurrence  $W[\alpha,\alpha+m-1]$ is inside the substring $s=(uv)^{i'+t+i''}u$. 
\item The occurrence $W[\alpha,\alpha+m-1]$ is at the right of the occurrence $W[h,m-1]$. This is symmetric to the previous case and can be proved with a similar argument. 
\end{itemize}
\end{proof}
Let $p$ be decomposed as $p=yqz=y(uv)^{t'}uz$ meaning that it contains $c_{q,p}=t'-t+1$ occurrences of $q$, the last step in the matching consists in comparing $y$ against $y'=q'[h-i'g-|y|,h-i'g-1]=W[h-i'g-|y|,h-i'g-1]$ and compare $z$ against $z'=q''[i''g,i''g+m-1]=W[i''g+m,i''g+m+|z|-1]$. We now distinguish four cases: 
\begin{enumerate}
\item $y$ is a suffix of $uv$ but $z$ is not a prefix of $vu$. We require that $y=y'$, $z=z'$ and moreover that $c_{q,W}\geq c_{q,p}$. If the requirement is fulfilled then we have a single match of $p$ with the fragment $W[i''g+|z|,i''g+|z|+m-1]$. Otherwise we do not have any match. 
\item $z$ is a prefix of $vu$ but $y$ is not a suffix of $uv$. This case is symmetric to the previous one. We also require the same three conditions: $y=y'$,$z=z'$ and $c_{p,W}\geq c_{q,p}$. If the requirement is fulfilled, then we have a single match of $p$ with the fragment $W[h-i'g-|y|,h-i'g-|y|+m-1]$.
\item Neither $y$ is suffix of $uv$ nor $z$ is prefix of $vu$. In this case we require $y=y'$, $z=z'$ and moreover that $c_{q,W}=c_{q,p}$. If the requirement is fulfilled, we have a single match of $p$ with the fragment $W[h-i'g-|y|,h-i'g-|y|+m-1]$, otherwise we have no match.
\item $z$ is a prefix of $vu$ and $y$ is suffix of $uv$. In this case we require that $c_{q,W}\geq c_{q,p}$. Then in case both $y=y'$ and $z=z'$ we conclude that we have $c_{q,W}-c_{q,p}+1$ matches where first match is $W[h-i'g-|y|,h-i'g-|y|+m-1]$ and the last match is $W[i''g+|z|,i''g+|z|+m-1]$ (the matches ). Otherwise we have $c_{q,W}-c_{q,p}$ matches:
\begin{itemize}
\item In case $y=y'$ but $z\neq z'$, the first match is $W[h-i'g-|y|,h-i'g-|y|+m-1]$ and the last one is $W[(i''-1)g+|z|,(i''-1)g+|z|+m-1]$.
\item In case $z=z'$ but $y\neq y'$, the first match is $W[h-(i'-1)g-|y|,h-(i'-1)g-|y|+m-1]$ and the last one is $W[i''g+|z|,i''g+|z|+m-1]$.
\end{itemize}
Note that any two consecutive occurrences are separated by $|uv|$ characters and reporting all occurrences takes time linear in the number of occurrences. 
\end{enumerate}
\begin{lemma}
The algorithm above correctly computes the occurrences of $p$ in $W$.
\end{lemma}
\begin{proof}
By lemma~\ref{lemma:substr_count} we can not have more than $c_{q,W}-c_{q,p}+1$ occurrences of $p$ in $W$. Thus to have at least a match we require that $c_{q,W}\geq c_{q,p}$. Now consider the starting position $\alpha$ of an occurrence of $p$ in $W$. The only possible values of $\alpha$ are of the kind $h-i'g-|y|,h-i'g-|y|+|uv|,,,h-i'g-|y|+|uv|^{c_{q,W}-c_{q,p}}$. If $c_{q,W}=c_{q,p}$ we conclude that we have a single possible match and this match is handled by the case 1, where the match is verified by matching the substrings $z$ and $y$ against the substrings $z'$ and $y'$. Now in case $c_{q,W}-c_{q,p}>0$, we could have more than one match. This case is handled by cases 2,3 and 4 in the algorithm above. 
We now prove that those 3 cases work correctly: we divide the set the matches into three categories, the leftest match, the rightest match and the ${q,W}-c_{q,p}-1$ middle matches. It can easily be verified that the middle matches are only possible in case $z$ is a prefix of $vu$ and $y$ is suffix of $uv$. It can also be easily verified that leftest match is only be possible if $y=y'$ and that $z$ is prefix of $vu$. Likewise the rightest match is only possible if $z=z'$ and that $y$ is suffix of $uv$. It can easily be checked that the three last cases correctly account for the matches. 
\end{proof}
%We emphasize that in fact all the factors which appear more than one time are in fact conjugate of each other. 
%Let's now consider the first and last occurrences of $q$ in $p$. Those two occurrences are separated by a distance $|q|/2$ which is multiple of $|uv|$, the (shortest) period of $p$. 
\paragraph{Implementation}
We now analyze in detail the data structures needed for the matching. The first step of the matching is to match the string $q=W[h,m-1]$ against all factors of $p$ of length $m-h$. Thus what we need is to build a dictionary on the set of factors of $p$ of length $m-h$. In addition, in case $q$ occurs as factor in $p$, we must proceed to a second step. In this second step we need to determine the factors $u,v,y,z$ of $p$. 
Thus the required dictionary must fulfill the following needs: take only space $O(m\log m)$ bits, answer in optimal $O(m\frac{\log\sigma}{w})$ and must return the necessary information to deduce the factors $u,v,y,z$ of $p$. 
For our implementation we can use a very simple data structure: we store a perfect hashing function so that for each factor $q$ we associate a triplet $(\alpha_q,\beta_q,r_q)$, where $\alpha_q$ is a pointer to first occurrence of $q$ in $p$, $\beta_q$ is number of occurrences of $q$ in $p$ and $r_q$ is the period of $q$ (note that we need $r_q$ only in case $\beta_q\geq 2$). For implementing the first step we can just compared $q$ with the factor $p[\alpha_q,\alpha_q+m-h-1]$. For implementing the second step the factors $u,v,y,z$ of $q$ are also factors of $p$ and their positions in $p$ are easily deduced as combinations of the parameters $\alpha_q$ $\beta_q$ and $r_q$. 
\paragraph{Running time analysis}
We now analyze more precisely the running time of the matching. We can prove that all the running time of the matching on the window $W$ take time $O(m\frac{\log\sigma}{w}+occ)$ where $occ$ is the number of reported occurrences. First, notice that in any case we are doing a constant number of operations among the following ones:
\begin{itemize}
\item Matching $q$ against the set $P$. 
\item Longest suffix repetition search for $uv$ in $q'$.
\item Longest prefix repetition search of $vu$ in $q''$. 
\item Comparing $z$ with $z'$ and $y$ with $y'$. 
\end{itemize}
Now the matching of $q$ can be done in optimal $O(|q|\frac{\log\sigma}{w})=O(m\frac{\log\sigma}{w})$ time by means of perfect hashing. The longest suffix repetition matching of $uv$ in $q'$ takes $O(|q'|\frac{\log\sigma}{w})=O(m\frac{\log\sigma}{w})$. Likewise doing a longest prefix repetition matching of $vu$ in $q''$ takes $O(m\frac{\log\sigma}{w})$. Finally comparisons of $z$ with $z'$ and of $y$ with $y'$ takes $O(m\frac{\log\sigma}{w})$ as the compared strings are of length $O(m)$. Thus we have proved the following lemma:
\begin{lemma}
The matching of all occurrences of a string $p$ of length $m$ into a string $W$ of length $m+h$ where $h\leq \lfloor m/3\rfloor$ can be done in optimal $O(m\frac{\log\sigma}{w}+occ)$ time, where $occ$ is the number of occurrences.
\end{lemma}

\subsection{Construction of the data structure}
The dictionary described above can easily be constructed in $O(m)$ time. 
%Additionally we need a global variable $period$ which represents the period of the factors which appear at least one time.
\paragraph{Determining the factors}
For determining the triplet $(\alpha_q,\beta_q,r_q)$ associated with each factor $q$, we begin by building the suffix tree of $p$. Then we do a DFS 	traversal of the suffix tree but in which the traversal is restricted to nodes of depth at most depth $m-h$. More precisely during the DFS traversal, at each time we reach a node $n_q$ whose path is of length at least $m-h$ whose path is prefixed by factor $q$, we know that all the nodes in the subtree of $n_q$ will be have pointers to all occurrences of $q$. Then it suffices to count the size of the subtree of $n_q$ to get $\beta_q$. For determining $\alpha_q$ and $r_q$, it suffices to traverse the pointers and select the two smallest ones $\alpha_q$ and $\gamma_q$. Then if $\beta_q\geq 2$ we will have that $r_q=\alpha_q-\gamma_q$. That is in case $q$ occurs two times then the period of $q$ is just the difference between the pointers to first and second occurrence of $q$. In order to state that factors which are equal, we use a temporary table $ID[1..n]$ which is set during the traversal of the suffix tree as well as a table $T$ of triplets. More precisely we attribute consecutive identifiers starting from zero to different patterns during the traversal of the suffix tree when we have a pattern $q$ we give it identifier $i$ and for each occurrence of $q$ at position $j$ we set $ID[j]=i$ and also set $T[j]=(\alpha_q,\beta_q,r_q)$.
\paragraph{Construction of the perfect hash function}
We now turn our attention to the creation of the perfect hash function. The construction follows a standard procedure, that is first the set $P$ is mapped injectively into a set of integers using a string hash function which transforms each string in $P$ into an integer of $O(\log n)$ bits, followed by the computation of the perfect hash function on that set. For implementing the first step we do $B=\lfloor w/\log\sigma\rfloor$ traversals of the pattern $p$. At each traversal $i$ we will compute the hash values for the patterns which start at positions $i,i+w/\log\sigma,i+2w/\log\sigma...$.  
At each traversal we compute the factor hash values in the following way: we conceptually divide each factor into blocks of size $B$ except the last block potentially of size less than $B$ which is padded with zero characters. Then we can compute the hash values using a polynomial hash function~\cite[section 5]{DGMP} $H_\gamma$ which is parametrized with a sufficiently large prime $\gamma$. That is we just use a table $H$ in which we initialize $H[0]=0$ and then successively for increasing $j$ set $H[i]=H[0]\otimes \gamma\oplus p[i+jB,(j+1)B]$ (except at the last step where we set $H[i]=H[0]\otimes \gamma\oplus p[i+jB,m-1]$). Then the hash value of associated with the factor occurring at position $i+jB$ will be given by $H[j+m/B]-H[j]/\gamma^j$. 
Then before doing the second step, we check whether the computed hash values for different factors (two factors $T[i_0]$ and $T[i_1]$ are considered distinct if their corresponding triplets $T[i_0]$ and $T[i_1]$ differ) are all distinct (if the hash function $H_\gamma$ is injective on the set $P$) and if it is not the case, redo the computation of the hash values using a newly chosen hash function $H_\gamma$ parametrized by a new randomly chosen prime $\gamma$. This procedure is repeated until we get a set of  distinct integer in which case we can proceed with the second step of the computation which consists in building the perfect hash function on the set of integers obtained in the first step. 
% We first compute the hash values for each substring of length $w/\log\sigma$ of the string $p$. Then compute all the $w/\log\sigma$ partial sums (or partial xors) of all consecutive blocks of sizes $w/\log\sigma$ which begin at offset $0,1,,,w/\log\sigma-1$. This also takes $O(m)$ time. Finally, compute the hash value for each factor.
%\subsection{}
\subsection{Tabulation based solution}
We now give a proof of theorem~\ref{theorem4}. In the case $m\geq \alpha/2$, theorem~\ref{theorem3} already gives the required query time using just $O(m\log m)$ bits of space. Thus we will focus on the case $m<\alpha/2$. For this case there is a very simple solution: consider that at each step $i$ (where $i$ is initialized as zero) we match the substring $s=T[i\alpha/2,i(\alpha/2)+\alpha-1]$ against all occurrences of the pattern $p$ (consider w.l.o.g that $\alpha$ is an even number) which start anywhere inside $s[0,\alpha/2-1]=T[i\alpha/2,(i+1)(\alpha/2)-1]$. Using the four russian technique we can for every possible string $s$ of length $\alpha$, all positions of all occurences of $p$ in $s$ (if any) which start at any position in $s[0,i\alpha/2-1]$. Note that every such occurrence must end at any position in $s[\alpha/2-1,\alpha-1]$. As we have $\sigma^\alpha$ strings of length $\alpha$ and we can have at most $\alpha/2$ positions of occurrences, the space used by the table which stores all those occurrences for every possible string is thus just $O(\sigma^\alpha\log\alpha)$ bits. 
\section{Conclusion}
In this paper, we have proposed four solutions to the problems of single and multiple pattern matching on strings in the RAM model. In this model we assume that we can read $\Theta(w/\log\sigma)$ consecutive characters of any string in $O(1)$ time. The first and third solutions have a query time which depends on the length of the shortest pattern (for multiple string matching) and the length of the only pattern respectively, in a way similar to that of the previous algorithms which aimed at average-optimal expected performance (not worst-case performance as in our case). Those two solutions achieve optimal query times if the shortest pattern (or the only pattern in the third solution) is sufficiently long. The second and fourth solutions have no dependence on the length of the shortest pattern but need to use additional precomputed space. They are interesting alternatives to the previous tabulation approaches by Bille~\cite{B09} and Fredriksson~\cite{F02}.
%\\We find it surprising that for $m<\frac{w}{\log\sigma}$ the query time of the algorithm is actually , we essentially get a matching time $O(n/m)$ even without using any tabulation and thus using space just $O(m)$ words. On the other hand, average optimal algorithms need $O(n\frac{\log_\sigma m}{m})$ time to match strings. Thus as long as $m\leq \frac{w\log w}{\log^2\sigma}$ the new  algorithm is faster than the average optimal non bit-parallel algorithms.
\\
This paper gives rise to two interesting open problems: 

\begin{itemize}
\item In order to obtain any speedup we either rely on the length of the shortest pattern being long enough (theorem~\ref{theorem1} and~\ref{theorem3}) or have to use additional precomputed space (theorems~\ref{theorem2} and~\ref{theorem4} ). An important open question is whether it is possible to obtain any speedup without relying on any of the two assumptions.
\item The space usage of both solutions is $\Omega(m\log m)$ bits, but the patterns themselves occupy just $(m\log \sigma)$ bits of space. The space used is thus at least a factor $\Omega(\log_\sigma m)$ larger than the space occupied by the patterns. An interesting open problem is whether it is possible to obtain an acceleration compared to the standard AC automaton while using only $O(m\log\sigma)$ bits of space.
\end{itemize} 
\section*{Acknowledgements}
The author wishes to thank Mathieu Raffinot for his many helpful comments and suggestions on a previous draft and to two anonymous reviewers for their helpful remarks and corrections on a previous version of the paper. 
\small 
\bibliography{Fast_Aho_Corasick} 
\normalsize
%\newpage
%\appendix
%\section{Proof of theorem 1}
%\subsection{Simulating transitions}
%\label{appendix:simul_trans}

%\subsection{Identifying matching occurrences}
%\label{appendix:ident_occ}

%\section{Proof of theorem 2}
%\label{appendix:tabulation_proof}

\end{document}